\documentclass[nohyperref]{article}
\usepackage{microtype}
\usepackage{graphicx}
\usepackage{subfigure}
\usepackage{booktabs}

\usepackage{algorithmic}
\usepackage{algorithm2e}

\usepackage[sort, numbers]{natbib}
\setcitestyle{square}

\usepackage{amsmath}
\usepackage{amssymb}
\usepackage{mathtools}
\usepackage{amsthm}

\usepackage[capitalize, noabbrev]{cleveref}
\usepackage{tikz}
\usetikzlibrary{arrows, automata}

\theoremstyle{plain}
\newtheorem{theorem}{Theorem}[section]

\theoremstyle{definition}
\newtheorem{definition}[theorem]{Definition}

\theoremstyle{remark}

\newcommand{\indep}{\perp \!\!\! \perp}
\newcommand{\MM}{\mathcal{M}}

\newcommand{\U}{\mathbf{U}}

\newcommand{\V}{\mathbf{V}}
\newcommand{\C}{\mathbf{C}}

\renewcommand{\u}{\mathbf{u}}
\newcommand{\X}{\mathbf{X}}
\newcommand{\x}{\mathbf{x}}

\renewcommand{\c}{\mathbf{c}}
\newcommand{\CC}{\mathcal{C}}
\newcommand{\DD}{\mathcal{D}}
\newcommand{\GG}{\mathcal{G}}

\author{Jason Saporta}
\title{Credible Bounds for Causal Quantities with Continuous Outcomes}
\date{%
    Department of Statistics, Iowa State University, Ames, Iowa, USA\\
    jsaporta@iastate.edu
}
\begin{document}
\maketitle
\begin{abstract}
  The problem of partial identification concerns bounding causal quantities that remain unidentifiable given the observed distribution and the causal diagram of the underlying structural causal model (SCM). While bounds have been developed for certain causal quantities with continuous outcomes, they are necessarily univariate and are unable to reflect differing levels of confidence. Building on previous work in partial identification, we propose a simple Bayesian method for deriving probabilistic bounds on a large class of causal quantities with potentially multivariate and/or continuous outcome variables.
\end{abstract}
\begin{center}
\textbf{Acknowledgements}
\end{center}

This work was partially funded  by the Center for Statistics and Applications in Forensic Evidence (CSAFE) through Cooperative Agreements 70NANB15H176 and 70NANB20H019 between NIST and Iowa State University, which includes activities carried out at Carnegie Mellon University, Duke University, University of California Irvine, University of Virginia, West Virginia University, University of Pennsylvania, Swarthmore College and University of Nebraska, Lincoln.

\begin{center}
\textbf{Keywords}
\end{center}

Bayesian analysis, causal inference, partial identification
\clearpage
\section{Introduction}
\label{intro}
Inferring causal relationships from observational data is a central task in the modern study of causality \cite{causality2e}. For example, given a binary treatment variable $X$ and an outcome variable $Y$, the average causal effect of $X$ on $Y$ is given by \[ACE[X \rightarrow Y] = E[Y | do(X = 1)] - E[Y | do(X = 0)].\]
This quantity might be used, for example, to establish the benefit of a new drug ($X = 1$) on a health outcome ($Y$) compared to a placebo ($X = 0$).

The $ACE$ and other causal quantities depend on \textit{interventional distributions} such as $p(y | do(x))$, representing the distribution of $Y$ when $X$ is fixed at a particular value $x$, regardless of the surrounding circumstances. This, in general, is not equal to the corresponding conditional distribution
\[ p(y | do(x)) \neq p(y | x) \]
due to potential unobserved confounding variables that influence both $X$ and $Y$. Details about the meaning of interventional distributions are covered in Section \ref{bg}.

Unlike conditional distributions, interventional distributions are not always uniquely determined (called \textit{identifiable} in the causality literature) by the joint distribution $p(x, y)$. Even if we are also given a \textit{causal diagram} (see Section \ref{bg}), $p(y | do(x))$ may remain unidentifiable. Algorithms based on the \textit{do}-calculus and other methods are now able to efficiently ascertain whether these probabilities are identifiable given a causal diagram and joint distribution. \cite{10.5555/3020419.3020446}

Even when interventional distributions are not identifiable from observational data, we can still draw samples from them using \textit{experiments}, where the effects of potential confounding variables are eliminated through randomized treatment assignment and a controlled setting. In an experiment, therefore, the interventional and conditional distributions are equivalent, and we may derive causal quantities using standard probabilistic and statistical tools. Experiments may be prohibitively expensive, unethical, or otherwise infeasible to carry out, and therefore do not completely solve the problem of unidentifiability.

The problem of \textit{partial identification} concerns bounding unidentifiable causal quantities. There are two common approaches to this: Finding ``hard'' bounds through optimization, and finding probabilistic ``soft'' bounds using a Bayesian approach.

The hard bounding approach expresses the causal quantity of interest, such as $E[Y | do(X = x)]$ as a function of unobservable probabilities (corresponding to possible confounding variables); this function is then minimized (maximized) under constraints imposed by the observed data to produce a lower (upper) bound on the causal quantity. This is usually achieved using a linear optimization solver \cite{doi:10.1080/01621459.1997.10474074}. Hard bounds have been derived for causal effects with discrete interventional variables and arbitrary outcome variables \cite{zhang2021bounding}; for specific classes of models, bounds are also available under continuous interventions \cite{NEURIPS2020_e8b1cbd0}.

This approach has a number of weaknesses. (1) Though it produces provably tight bounds in the infinite-data limit, the finite-sample statistical properties of these bounds are not well understood. (2) Due to the nature of the optimization problem, the method only provides bounds for one scalar outcome variable at a time; there is no concept of a truly multidimensional bounding region. Even to get a ``rectangular'' region (ignoring how outcome variables covary), a separate run of the optimizer is necessary for each dimension. (3) Finally, there is no way to scale the resulting bounds to account for different levels of confidence or belief; there is no analogue of a 95\% confidence interval, for instance.

By contrast, the Bayesian soft bounding approach \cite{cp1997} entails putting a prior distribution on the latent variables and parameters of a causal model, then sampling from the posterior predictive interventional distribution $p(y | do(x), \DD)$. From there, bounds can be derived using credible sets as per standard Bayesian practice \cite{gelman2013bayesian}. The Bayesian approach has been used to derive bounds for counterfactual queries on graphs of arbitrary complexity, with the caveat that all observed variables be discrete \cite{zhang2021pi}.

Our contribution is to provide a unified posterior sampling method covering a large class of unidentifiable interventional distributions, from which credible bounds on many causal quantities can be derived. In particular, our method can sample from (possibly multivariate) interventional distributions with continuous outcome variables. The method is based on the novel concept of canonicalizable components, which allows us to leverage the canonical representation \cite{eoci} for parts of an SCM, but not necessarily the whole thing. To the best of our knowledge, this is the first time this has been achieved across such a wide range of causal models.

\subsection{Background}
\label{bg}
A \textbf{Structural Causal Model} $\MM = \langle \U, \V, F, P_\U \rangle$ \cite{causality2e} consists of a finite set of independent ``exogenous'' random variables $\U$ defined on an arbitrary measurable space with distribution $P_\U$, a finite set of derived ``endogenous'' variables $\V$, and a set of functions $F = \{f_V: V \in \V\}$ defined such that
\[ V = f_V(Pa_V, U_V) \]
for ``parents'' $Pa_V \subseteq \V \setminus V$ and $U_V \subseteq \U$. We will assume that all SCMs discussed here are \textbf{acyclic} and define a (not necessarily unique) ordering over $\V$ such that, given $\U = \u$, we are able to uniquely compute the value of each variable $V \in \V$ according to that ordering. This is called a \textbf{topological ordering} of $\V$ and can be found from an SCM with standard algorithms \cite{clrs}.

Whereas an SCM uniquely defines a probability distribution $P_\V$ over the endogenous variables, the converse is not true \cite{pch}. SCMs also uniquely define interventional distributions. If $X, Y \in \V$, for example, we can sample from $p(y | do(X = x))$ by first replacing $f_X$ with the interventional value $x$, then sampling from $\V$ as described above until we arrive at $Y$. With a fully-specified SCM, all causal quantities are therefore identifiable.

However, SCMs are almost never considered to be completely specified. A particularly important source of difficulty when working with SCMs is that $\U$ is usually considered to be an unobservable random variable with an unknown distribution on an arbitrary measurable space. Because of this, it is inappropriate to associate computationally convenient distributions with $\U$ without additional thought; a Gaussian ``prior'' distribution, for example, assigns zero probability to a large class of SCMs that may be compatible with the observed distribution $p(\mathbf{v})$ and causal diagram $\GG$. Therefore, even a seemingly uninformative prior like a Gaussian with large variance in fact encodes vast amounts of information and must be justified. Similar considerations apply to the set of functions $F$. We will consider this problem in more detail in Section \ref{cc}.

Given an SCM $\MM$, its \textbf{causal diagram} $\GG$ is an acyclic directed mixed graph (ADMG) constructed as follows:
\begin{enumerate}
  \item Find a topological ordering of $\V$.
  \item Iterating accordingly, for each $V$ add a node to $\GG$, and add a directed edge from each member of $Pa_V$ to $V$.
  \item For each pair of variables $(V_i, V_j)$, add a bidirected edge connecting $V_i$ and $V_j$ if the arguments of $f_{V_i}$ and $f_{V_j}$ share an element of $\U$.
\end{enumerate}
Intuitively, the bidirected edges of $\GG$ correspond to potential unobserved confounding between the observable variables $\V$.

Given an SCM $\MM$, \cite{TianJin2002Agic} defined a partition $\CC$ of $\V$ according to the following relation: $V_i \approx V_j$ if $i = j$ or there exists a bidirected path between $V_i$ and $V_j$ in the causal diagram $\GG$. (The relation $\approx$ is clearly reflexive, symmetric, and transitive.) Each member $\C \in \CC$ is called a \textbf{confounded component} (c-component). We will treat these as random vectors taking on values $\c \in \Omega_\C$ and having parents $Pa_\C = \bigcup_{V \in \C} Pa_V$ taking on values $pa_\C \in \Omega_{Pa_\C}$.

\section{Canonicalizable Components}
\label{cc}
In the standard partial identification problem setting, which we also follow here, we are provided with $N$ \textit{iid} samples from the observational distribution $p(\mathbf{v})$ together with a causal diagram $\GG$. We view these as reflecting an underlying SCM representing the full data-generating process, but as described earlier, we lack key information about $\U$ and $F$ necessary for a full specification.

If all observed variables $V \in \V$ have finite support, the \textbf{canonical representation} of an SCM $\MM$ \cite{eoci} allows us to view the different possible values for $\U$ as selecting for a set of functional assignment mechanisms for each $V$. That is, for a given $V = f_V(Pa_V, U_V)$, the value of $U_V$ does nothing but specify how values of $Pa_V$ get mapped to values of $V$. When all $V$ are finite, the number of possible functions specifying their values is finite as well.

Given an SCM $\MM = \langle \U, \V, F, P_\U \rangle$ and some $V \in \V$, define $F_V$ to be the set of functions mapping from the support of values of $Pa_V$ to values of $V$. If $Pa_V = \emptyset$, then interpret this set as just the support $\Omega_V$:
\[ F_V =  \left\{
\begin{array}{ll}
  \{ f: \Omega_{Pa_V} \mapsto \Omega_V \} & Pa_V \neq \emptyset \\
  \Omega_V & Pa_V = \emptyset \\
\end{array} 
\right. \]

Under the canonical representation, then, the previously intractable $\U$ can simply be viewed as having a categorical distribution; each of its possible values $\mathbf{u}$ then specifies which member of each $F_V$ is used to generate the values of $V$ from its parents.

Using this insight, \cite{cp1997} put a Dirichlet prior distribution on the probability vector $\pi$ underlying the distribution of $\U$, which essentially constituted a prior on the whole space of SCMs compatible with the causal diagram they were considering (see Figure \ref{fig1}). Using a Gibbs sampler, they approximated the posterior distribution $p(\pi | \DD)$, and were thereby able to approximate the posterior distribution of potentially unidentifiable causal quantities like the $ACE$.

\begin{figure}
  \centering
  \begin{tikzpicture}[->,>=stealth',shorten >=1pt,auto,node distance=2.8cm,
                  semithick,scale=1,transform shape]
    \node[state]    (A)                    {$Z$};
    \node[state]    (B) [right of=A]       {$D$};
    \node[state]    (Y) [right of=B]       {$Y$};
    
    \path (A) edge node {} (B)
          (B) edge node {} (Y)
          (Y) edge[dashed, <->] [bend right] node {} (B);
  \end{tikzpicture}
  \caption{The causal diagram of the ``partial compliance'' model \cite{cp1997}. All 3 variables are binary-valued. Utilizing the canonical representation and a Dirichlet prior, bounds on $ACE[D \rightarrow Y]$ were created by sampling from $p(y | do(d), \DD)$, for $\DD = \{(z_n, d_n, y_n): 1 \leq n \leq N \}$.}
  \label{fig1}
\end{figure}
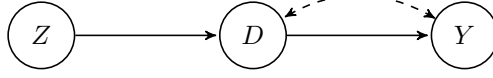

The previous discussion, however, only applies to situations where each $V$ has finite support, or equivalently, when $\sum_{V \in \V} |F_V| < \infty$. When this fails, there are an infinite number of possible functions linking the various variables, $\U$ cannot be viewed as having a categorical distribution, and we can no longer think about putting a Dirichlet prior on some underlying $\pi$.

By contrast, our approach is not to require that all $V$ be finite, but rather only those in the vicinity of each intervention variable $X$. These $V$ and their corresponding $U_V$ can then be ``canonicalized'' in a similar fashion to \cite{cp1997}. The remainder of the SCM can be treated without regard to possible confounders or other causal structure.

This will allow us to leverage the benefits of the canonical representation for approximating causal quantities while weakening the conditions required for its application.

\begin{definition}[Canonicalizable Component]
\label{def1}
  A c-component $\C \subseteq \V$ is called \textbf{canonicalizable} if $N_\C = \prod_{V \in \C} |F_V| < \infty$.
\end{definition}
Given a canonicalizable component $\C \in \CC$, there are $N_\C$ possible assignments of functions to variables $V \in \C$. Let $\mathbf{\pi_\C}$ be a probability vector of length $N_\C$, and define $U_\C\sim \textrm{Categorical}(\pi_\C)$. We can create a bijective correspondence between the values of $U_\C$ and these function assignments.

While prior work has discretized SCMs ``one c-component at a time'', it nonetheless assumed that the entire SCM can be converted to the canonical representation, and therefore had no need to differentiate between those components which can be made to follow the canonical representation and those which can't in a general SCM. \cite{zhang2021pi}

We now establish that c-components, canonicalizable or not, can in some sense be separated and considered independently from each other. This will allow us to approximate the posteriors of multiple $\pi_\C$ in parallel if there are multiple canonicalizable components.

\begin{theorem}
\label{thm1}
  Consider an SCM $\MM = \langle \U, \V, F, P_\U \rangle$ with associated causal diagram $\GG$. Let $\C \in \CC$ be a c-component of $\GG$, and define $\GG_\lambda$ as the graph $\GG$ with an added node $\lambda$, edges pointing from $\lambda$ to each $V \in \C$, and all bidirected edges connected to $V \in \C$ deleted.

Then given $\C \cup Pa_\C$, $\lambda$ is d-separated from all other variables $Z \not \in \C \cup Pa_\C$.
\end{theorem}
\begin{proof}
Consider a variable $Z \not \in \C \cup Pa_\C$, and consider any path between $\lambda$ and $Z$. This path must follow one of three forms for $V \in \C$, $W \in \C \cup Pa_\C$:
\begin{enumerate}
  \item $\lambda \rightarrow V \rightarrow W \rightarrow \ldots \leftrightarrow Z$
  \item $\lambda \rightarrow V \leftarrow W \leftarrow \ldots \leftrightarrow Z$
    \item $\lambda \rightarrow V \leftarrow W \rightarrow \ldots \leftrightarrow Z$
\end{enumerate}
Note that the $\leftrightarrow$ pointing to/from $Z$ is not a bidirected edge in the sense of causal diagrams, but rather represents either a $\leftarrow$ or a $\rightarrow$.

In Case 1, $\lambda \rightarrow V \rightarrow W$ constitutes a chain structure. $V \in \C$, so this path is blocked given $\C \cup Pa_\C$.

In Case 2, $V \leftarrow W \leftarrow \ldots$ is also a chain. Since $W \in \C \cup Pa_\C$, this path is also blocked.

Finally, in Case 3, $V \leftarrow W \rightarrow \ldots$ is a fork, and therefore such a path is blocked for the same reason as in Case 2.
\end{proof}

For a canonicalizable component $\C$, ``canonicalization'' amounts to creating a new node $\lambda = U_\C$ pointing to each $V \in \C$ and deleting the associated bidirected edges. If we now consider another c-component $\C'$ with associated latent categorical $U_{\C'}$, $U_\C$ and $U_{\C'}$ are guaranteed to be d-separated, and therefore independent, given $\V$. If we also create $\pi_\C$ and $\pi_{\C'}$ pointing to $U_\C$ and $U_{\C'}$, respectively, it's clear that $\pi_\C \indep \pi_{\C'} | \V$ and can be reasoned about in isolation.
\section{Posterior Predictive Sampling from Interventional Distributions}
\label{ppsftid}
\subsection{Preliminaries}
\label{prelim}
Consider a dataset $\DD = \{v_n: V \in \V, n \in 1,\ldots,N\}$ representing $N$ \textit{iid} samples from $p(\mathbf{v})$, and an associated causal diagram $\GG$. We view these inputs as resulting from an underlying, unseen SCM.

Let $\X \subseteq \V$, and consider a set of intervention values $\x \in \Omega_\X$. Because the set $\CC$ of c-components partitions $\V$, each $X \in \X$ is also a member of a c-component $\C_X \in \CC$. These are not necessarily unique; for $X, X' \in \X$, it's possible that $\C_X = \C_{X'}$. Let $\CC_\X = \{\C_X: X \in \X\}$.

We assume that each such $\C_X$ is canonicalizable as per Definition \ref{def1}. Though this limits the set of possible interventional distributions we can sample from, it is necessary in order to achieve truly uninformative priors on the latent confounders.

We now modify the causal diagram $\GG$ in order to obtain the benefits from Theorem \ref{thm1}:
\begin{itemize}
  \item For each $\C \in \CC_\X$, delete all bidirected edges and introduce a new $U_\C$ node with outgoing edges pointing to each $V \in \C$. Also introduce a node $\pi_\C$, with a single edge pointing to $U_\C$.
  \item For all other variables $V \not \in \bigcup_{X \in \X} \C_X$, delete all bidirected edges. Optionally, we can choose to add a new node $\theta$ pointing to these variables representing statistical model parameters equipped with a prior distribution. Treating these remaining $V$ as a single c-component, applying Theorem \ref{thm1} shows that $\theta$ can be considered in isolation from the various $\pi_{\C_X}$.
\end{itemize}
Denote this new graph as $\GG^*$. Note that $\GG^*$ has no bidirected edges and is now a standard DAG. Also note that outside the $\C_X$, a detailed specification of confounding associations is unnecessary.

\subsection{Sampling From the Posterior Mixture}
\label{mix}
Our goal is to sample from
\begin{align*}
  p(\mathbf{v} | do(\X = \x^*), \DD) = \int &p(\mathbf{v} | do(\X = \x^*), u_\X)\\
                                            &\cdot p(u_\X | \DD) d(u_\X),
\end{align*}
where $U_\X = \bigcup_{\C \in \CC_\X} U_\C$. Because each $\C_X$ is canonicalizable, each of these $U_\C$ is well-defined as described in Section \ref{cc}.

This is simply a mixture model; to sample from it, we first draw from the posterior distribution $p(u_\X | \DD)$, and then from the corresponding component $p(\mathbf{v} | do(\X = \x^*), u_\X)$. In order to do this, we first leverage Theorem \ref{thm1} to factor the posterior:
\begin{align*}
p(u_\X | \DD) &= \prod_{\C_X} p(u_{\C_X} | \DD)\\
              &= \prod_{\C_X} p(u_{\C_X} | \pi_{\C_X}) p(\pi_{\C_X} | \DD)\\
              &= \prod_{\C_X} p(u_{\C_X} | \pi_{\C_X}) \cdot \prod_{\C_X} p(\pi_{\C_X} | \DD) 
\end{align*}
Our first goal is to approximate each $p(\pi_{\C_X} | \DD)$; as noted before, this can be done in parallel. Algorithm \ref{pi-post} uses a Gibbs sampler to get $M$ draws from each of these posteriors after $B$ burn-in iterations; by conditional independence, this is equivalent to a draw from joint posterior over all $\pi_{\C_X}$.\footnote{Dependence between successive draws in a Gibbs sampler can be mitigated by determining the lag between two draws necessary for them to be approximately decorrelated.}

In Algorithm \ref{pi-post}, note that...
\begin{itemize}
  \item The \textit{mask} operation corresponds to what \cite{cp1997} call ``compatibility''; given $\c_n, pa_{\C, n}$, $U_{\C, n}$ cannot take on values corresponding to function assignments that disagree with the data. We capture this by using a bit mask to zero out probabilities in $\pi_\C$ corresponding to invalid assignments, then rescaling it to ensure it sums to $1$. Note that a small $\epsilon$ value may need to be added to each element of $\pi_\C^*$ before rescaling to ensure numerical stability; in practice this can be made small enough not to significantly alter the functioning of the algorithm.
  \item On each iteration, each $\pi_{\C_X}$ follows a Dirichlet distribution due to Dirichlet-categorical conjugacy.
\end{itemize}

\begin{algorithm}[tb]
  \caption{$\pi_\C | \DD$ Gibbs Sampler}
  \label{pi-post}
\begin{algorithmic}[1]
  \item[] {\bfseries Input:} $\c_{1,\ldots,N}, pa_{\C, 1,\ldots,N}$, $B$, $M$
  \STATE $\alpha \gets $vector of ones with length $N_\C < \infty$
  \STATE $\pi_\C \gets draw(Dirichlet(\alpha))$
  \FOR{$i=1$ {\bfseries to} $B + M$}
    \FOR{$n = 1$ {\bfseries to} $N$}
      \STATE $\pi^*_\C \gets mask(\c_n, pa_{\C, n}) \odot \pi_\C$
      \STATE $\pi^*_\C \gets \pi^*_\C / sum(\pi^*_\C)$
      \STATE $U_{\C, n} \gets draw(Categorical(\pi^*_\C))$
      \STATE $\alpha[U_{\C, n}] \mathrel{+}= 1$
    \ENDFOR
    \STATE $\pi_\C \gets draw(Dirichlet(\alpha))$
    \STATE $\alpha \gets $vector of ones with length $N_\C < \infty$
    \IF {$i > B$}
      \STATE Store $\pi_\C$
    \ENDIF
  \ENDFOR
  \STATE Return the $M$ sampled $\pi_{\C}$ vectors
\end{algorithmic}
\end{algorithm}

Now suppose that for each $X \in \X$, we have a sampled value for $\pi_{\C_X}$. We can then sample values for each corresponding $U_{\C_X}$. We can now move on to sampling from the mixture component $p(\mathbf{v} | do(\X = \x^*), u_\X)$. First, write out the factors corresponding to each $V \in \V$:
\begin{align*}
  p(\mathbf{v} | do(\X = \x^*), u_\X) =& \prod_{X \in \X} \mathbf{1}_{X = x^*} \prod_{V \not \in \bigcup_{\X} \C_X} p(v | pa_V)\\
                                             \cdot& \prod_{\C \in \CC_X} \prod_{V \in (\C \setminus X)} p(v | pa_V, u_\mathbf{C})
\end{align*}
Rearranging this according to a topological ordering of $\V$, we can sample from the conditional distribution of each $V \in \V$ sequentially. In particular, note that...
\begin{itemize}
  \item If $V \in X$, then its value has already been determined by the intervention.
  \item If $V \not \in X$, but is in some $\C_X$, then its distribution $p(v | pa_V, u_{\C_X})$ is simply a point mass at the value of the function corresponding to $u_{\C_X}$ and evaluated at $pa_V$.
  \item Otherwise, we simply draw $V$ according to its conditional distribution as determined by $p(\mathbf{v})$, which can be estimated from our $N$ draws via statistical methods.
\end{itemize}
This sampling scheme is summarized in Algorithm \ref{v-post}.

\begin{algorithm}[t]
  \caption{$\V | do(\x^*), u_\X$ Sampler}
  \label{v-post}
\begin{algorithmic}[1]
  \item[] {\bfseries Input:} $\{\pi_{\C_X}: X \in \X\}, \x^*, \GG^*$
  \STATE $u_\X \gets draw \left(\prod_{\C \in \CC_\X} p(u_\C | \DD) \right)$
  \FOR{$V$ {\bfseries in} $topologicalSort(V, \GG^*)$}
    \IF{$V \in \mathbf{X}$}
      \STATE $V \gets \textrm{the appropriate intervention value } x^*$
    \ELSIF{$V \in \C_X$ for some $X \in \mathbf{X}$}
      \STATE $V \gets f_V(pa_V, u_{\C_X})$
    \ELSE
      \STATE $V \gets draw(p(V | pa_V, \theta))$
    \ENDIF
  \ENDFOR
  \STATE Return the sampled vector $\mathbf{v}$
\end{algorithmic}
\end{algorithm}

Using this, for each of the $M$ sets of $\pi_{\C_X}$ values found with Algorithm \ref{pi-post}, we can draw $J$ samples from the corresponding SCM with Algorithm \ref{v-post}, and by the Law of Large Numbers get arbitrarily good approximations to causal quantities such as the average causal effect seen in Section 
\ref{intro}. These may be multivariate and depend on variables $V$ with infinite support.

\subsection{Bounds}
\label{bounds}
Suppose a quantity $Q$ is univariate and we have $M$ draws from its (approximate) posterior $p(q|\DD)$. A simple way to find a $(1 - \alpha) 100\%$ credible interval for $Q$ is to take those values that lie between the $\left(\frac{\alpha}{2}\right) 100\%$ and $\left(1 - \frac{\alpha}{2}\right) 100\%$ quantiles of the sample. This is called a \textbf{central credible interval} \cite{gelman2013bayesian}. Similar quantile-based intervals can also be derived.

For multivariate $\mathbf{Q}$, there is no analogue of a central credible interval. We can, however, find a $(1 - \alpha) 100\%$ \textbf{highest posterior density region}. This means finding a region $A \subseteq \Omega_\mathbf{Q}$ such that $P(\mathbf{Q} \in A | \DD) = 1 - \alpha$ and $A$ has minimum total volume (the HPD region may consist of several disjoint parts for multimodal distributions). Algorithms that achieve this tend to be more computationally demanding and mathematically involved compared to quantile-based intervals. \cite{10.2307/2684423}, \cite{10.1214/09-AOS766}
\section{Extensions}
\label{ex}
In addition to the standard ``hard'' interventions like $do(X = x)$, we can also consider the distribution of $Y$ under ``soft'' interventions such as assigning $X$ a new function of other variables in the model or a new conditional distribution \cite{correa2020calculus}, \cite{eoci}. Only minor modifications to Algorithm \ref{v-post} are necessary in order to sample from the corresponding interventional distributions:
\begin{itemize}
  \item To perform a functional intervention $do(X = f_X^*(pa^*_X))$, the ordering of $\V$ in Algorithm \ref{v-post} must comply with the post-intervention diagram, which is obtained by taking the original diagram, deleting the edges from $Pa_X \rightarrow X$ and adding the edges from $Pa^*_X \rightarrow X$. $X$ is then assigned the value of $f_X^*(pa^*_X)$ in Algorithm \ref{v-post} at the appropriate iteration.
  \item A stochastic intervention $do(X \sim p^*(x | pa^*_X))$ operates just as a functional intervention, except $X$ is drawn from $p^*$ rather than assigned a value according to a function.
\end{itemize}
In both cases, the requirement that $\C_X$ be canonicalizable in the original causal diagram remains unchanged. In addition, the post-intervention graphs must remain acyclic. The posterior distributions $p(\pi_\C | \DD)$ are learned as before, and the support $\Omega_X$ of $X$ cannot expand to include new values not possible for $X$ in the pre-intervention distribution.
\section{Experiments}
\label{exps}
\subsection{Partial Compliance $+$ Continuous Noise}
\label{pccn}
\begin{figure}
  \centering
  \begin{tikzpicture}[->,>=stealth',shorten >=1pt,auto,node distance=2.8cm,
                      semithick,scale=0.8,transform shape]
    \node[state]    (Z)                    {$Z$};
    \node[state]    (D) [right of=Z]       {$D$};
    \node[state]    (Y) [right of=D]       {$Y$};
    \node[state]    (W) [right of=Y]       {$W$};
    
    \path (Z) edge node {} (D)
          (D) edge node {} (Y)
          (Y) edge[dashed, <->] [bend right] node {} (D)
              edge node {} (W);
  \end{tikzpicture}
  \caption{The causal diagram used for Experiment 1. $Z$, $D$, and $Y$ are binary, $W$ is continuous. Compare to Figure \ref{fig1}.}
  \label{ex1cd}
\end{figure}
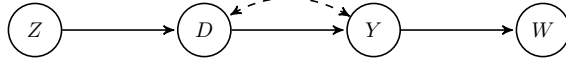

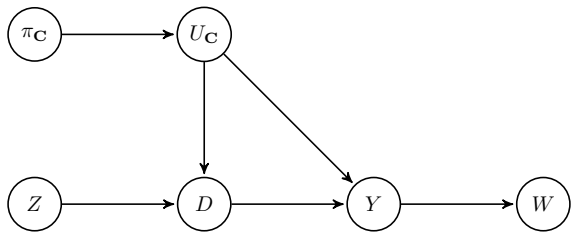
\begin{figure}
  \centering
  \begin{tikzpicture}[->,>=stealth',shorten >=1pt,auto,node distance=2.8cm,
                      semithick,scale=0.8,transform shape]
    \node[state]    (Z)                    {$Z$};
    \node[state]    (D) [right of=Z]       {$D$};
    \node[state]    (Y) [right of=D]       {$Y$};
    \node[state]    (W) [right of=Y]       {$W$};
    \node[state]    (U) [above of=D]       {$U_\C$};
    \node[state]    (P) [above of=Z]       {$\pi_\C$};
    
    \path (Z) edge node {} (D)
          (D) edge node {} (Y)
          (Y) edge node {} (W)
          (U) edge node {} (D)
              edge node {} (Y)
          (P) edge node {} (U);
  \end{tikzpicture}
  \caption{The modified causal diagram $\GG^*$ used for Experiment 1.}
  \label{ex1mcd}
\end{figure}
Figure \ref{fig1} shows the causal diagram of the partial compliance model used by \cite{cp1997}. It consists of three binary variables, organized into two canonicalizable c-components. In order to sample from the posterior distribution of the unidentifiable $ACE[D \rightarrow Y]$, they proposed the equivalent of Algorithm \ref{pi-post} with $\X = \{D\}$, so $\C_X = \{D, Y\}$.

The distribution in Table \ref{exp1table} to validate the approach. This distribution is constructed in such a way that $ACE[D \rightarrow Y]$ is in fact identifiable; it has the value $0.55$.

\begin{table}[b]
  \caption{$p(z, d, y)$ in Experiment 1}
  \label{exp1table}
  \vskip 0.15in
\begin{center}
\begin{small}
\begin{sc}
\begin{tabular}{l|l|l|l}
\toprule
  $Z$ & $D$ & $Y$ & $p(\cdot)$\\
\midrule
  $0$ & $0$ & $0$ & $0.275$ \\
  $0$ & $0$ & $1$ & $0.0$ \\
  $0$ & $1$ & $0$ & $0.225$ \\
  $0$ & $1$ & $1$ & $0.0$ \\
  $1$ & $0$ & $0$ & $0.225$ \\
  $1$ & $0$ & $1$ & $0.0$ \\
  $1$ & $1$ & $0$ & $0.0$ \\
  $1$ & $1$ & $1$ & $0.275$ \\
\bottomrule
\end{tabular}
\end{sc}
\end{small}
\end{center}
\vskip -0.1in
\end{table}

We extend the partial compliance model with the simplest possible addition that demonstrates our contribution. Figure \ref{ex1cd} shows the same model, but with an additional variable $W$. We suppose that $W \sim Normal(Y, 10^{-4})$. That is, $W$ is just $Y$ with a negligible amount of continuous noise. Clearly, the true $ACE[D \rightarrow W] = ACE[D \rightarrow Y] = 0.55$.

Figure \ref{ex1mcd} shows the modified diagram $\GG^*$. Because $p(z)$ and $p(w | y)$ are both known in this example, it is not necessary to learn additional parameters $\theta$; they have therefore been excluded from the graph.\footnote{Alternatively, we can imagine $\theta$ being added to the graph as described in Section \ref{prelim}. Because $p(z)$ and $p(w | y)$ don't share parameters, we can then factor that $\theta$ into two nodes, $\theta_1$ and $\theta_2$, pointing to $Z$ and $W$, respectively. Finally, because the distributions are known in this example, we would say that $p(\theta_1)$ and $p(\theta_2)$ are both point masses at the true parameter values.}

Note that $\C_D \cup Pa_{\C_D} = \{Z, D, Y\}$; by Theorem \ref{thm1}, only these variables are necessary to draw from the posterior $U_{\C_D} | \DD$. While there are jointly only 8 possible values for these three variables, $N_\C = |F_D| \cdot |F_Y| = 16$, so our prior distribution on $\pi_\C$ is a 16-dimensional Dirichlet with parameter vector $\mathbf{1}$.

Figure \ref{ex1plot} shows the approximate posterior density of $ACE[D \rightarrow W]$. We used $N = 1000$ \textit{iid} simulated data points, $B = 5000$ burn-in iterations, and $M = 10^5$ draws from (the MCMC approximation to) $p(U_\C | \DD)$. $J = 10^5$ samples were drawn from each of $p(w | do(D = 0))$ and $p(w | do(D = 1))$; the sampled values were averaged to produce $10000$ draws from the posterior $p(ACE[D \rightarrow W] | \DD)$; the result was then smoothed with a kernel to generate the plot.

\begin{figure}[t]
  \centering
  \includegraphics[width=\textwidth]{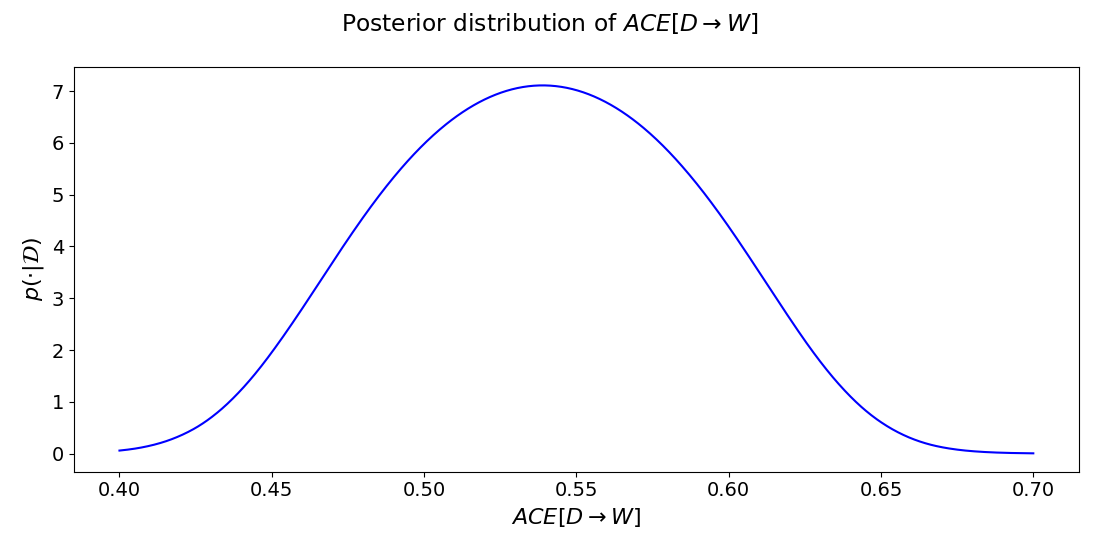}
  \caption{Posterior density of $ACE[D \rightarrow W]$. Whereas the prior distribution $\pi_\C \sim Dirichlet(\mathbf{1})$ induces a prior density on the $ACE$ that is symmetric about $0$, the posterior places significantly more probability mass near the true value.}
  \label{ex1plot}
\end{figure}

Unsurprisingly, most of the probability mass is located near the true $ACE$ value of $0.55$; either of the methods described in Section \ref{bounds} will capture the true value for any reasonable choice of $\alpha$.

\subsection{A More Complicated Model}
\label{complicated}
\begin{figure}
  \centering
  \begin{tikzpicture}[->,>=stealth',shorten >=1pt,auto,node distance=2.8cm,
                      semithick,scale=0.8,transform shape]
    \node[state]    (X)                    {$X$};
    \node[state]    (Z1) [below of=X]      {$Z_1$};
    \node[state]    (Z2) [right of=X]     {$Z_2$};
    \node[state]    (Y)  [right of=Z1]     {$Y$};
    
    \path (X)  edge node {} (Z1)
          (Z1) edge node {} (Y)
          (Z2) edge node {} (Y)
          (X) edge[dashed, <->] node {} (Z2)
          (Z1) edge[dashed, <->] node {} (Z2);
  \end{tikzpicture}
  \caption{The causal diagram $\GG$ used for Experiment 2. $p(y | do(x))$ is unidentifiable. \cite{causality2e}}
  \label{ex2cd}
\end{figure}
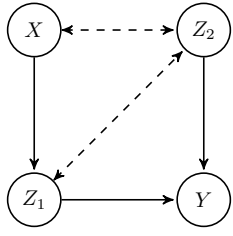

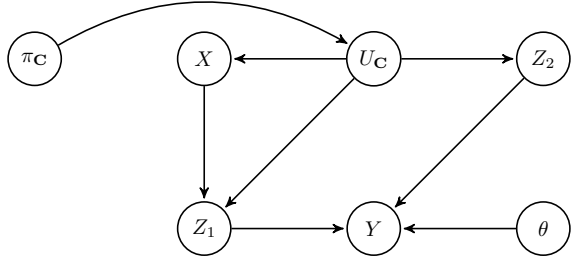
\begin{figure}
  \centering
  \begin{tikzpicture}[->,>=stealth',shorten >=1pt,auto,node distance=2.8cm,
                      semithick,scale=0.8,transform shape]
    \node[state]    (X)                    {$X$};
    \node[state]    (Z1) [below of=X]      {$Z_1$};
    \node[state]    (U)  [right of=X]      {$U_\C$};
    \node[state]    (Y)  [right of=Z1]     {$Y$};
    \node[state]    (Z2) [right of=U]      {$Z_2$};
    \node[state]    (t)  [right of=Y]      {$\theta$};
    \node[state]    (P)  [left of=X]      {$\pi_\C$};
    
    \path (X)  edge node {} (Z1)
          (Z1) edge node {} (Y)
          (t)  edge node {} (Y)
          (U)  edge node {} (X)
          (U)  edge node {} (Z1)
          (U)  edge node {} (Z2)
          (Z2)  edge node {} (Y)
          (P) edge [bend left] node {} (U);
  \end{tikzpicture}
  \caption{The modified causal diagram $\GG^*$ used for Experiment 2.}
  \label{ex2mcd}
\end{figure}
Figure \ref{ex2cd} shows a causal diagram from \cite{causality2e} where $p(y | do(x))$ is unidentifiable. We can divide this diagram into two c-components: $\CC = \{\{X, Z_1, Z_2\}, \{Y\}\}$.

We consider the case where $X$, $Z_1$, and $Z_2$ are all binary, and $Y | z_1, z_2 \sim Normal(\theta_{z_1, z_2}, 1)$, with \[\theta = \begin{bmatrix}
1 & 0\\
1 & 1.5
\end{bmatrix}. \]
Specifically, training data were generated according to the following process:
\begin{enumerate}
  \item Draw $U_1 \sim Beta(0.55, 0.45)$
  \item Draw $U_2 \sim Beta(0.35, 0.65)$
  \item $X \gets \mathbf{1}_{U_1 > 0.5}$
  \item $Z_1 \gets \mathbf{1}_{U_2 < 0.6} \oplus X$
  \item $Z_2 \gets \mathbf{1}_{U_1 < 0.4} \oplus \mathbf{1}_{U_2 > 0.3}$
  \item Draw $Y \sim Normal(\theta_{z_1, z_2}, 1)$
\end{enumerate}
In this example, we aim to get a $95\%$ bound for the quantity $E[Y | do(X = 0)]$. We generated $N = 1000$ \textit{iid} training examples $(x, z_1, z_2, y)$ using the process described above. $\C = \{X, Z_1, Z_2\}$ is canonicalizable, and therefore it's possible to sample from $p(y | do(x))$. We start by creating the modified diagram $\GG^*$ corresponding to our problem, which is shown in Figure \ref{ex2mcd}.

$N_\C = |F_X| \cdot |F_{Z_1}| \cdot |F_{Z_2}| = 2 \cdot 4 \cdot 2 = 16$, so there are 16 possible values for $U_\C$, and a corresponding 16-dimensional $Dirichlet(\mathbf{1})$ prior on $\pi_\C$. We treat $\theta$ in the model as a random $2 \times 2$ matrix (zero-indexed); we put independent priors on each of these: $\theta_{z_1, z_2} \sim Normal(0, 1)$.

We can now use Algorithm \ref{pi-post} to get $M = 1000$ draws from $p(\pi_\C | \DD)$. The elements of $\theta$ are all easily updated to posterior distributions due to prior conjugacy; therefore, we can also easily draw $M$ samples from $p(\theta | \DD)$. Now, we follow Algorithm \ref{v-post} to get $J = 1000$ draws from $p(y | do(X = 0))$ for each combination of $pi_\C$ and $\theta$.

The result is shown in Figure \ref{ex2plot}, and we see an associated 95\% central credible bound of $[0.81, 1.20]$, which contains the true value of $1.00$\footnote{The ``true value'' is actually also unknown, but was estimated with $10^4$ separate samples from the interventional distribution.}.

\begin{figure}[t]
  \centering
  \includegraphics[width=\textwidth]{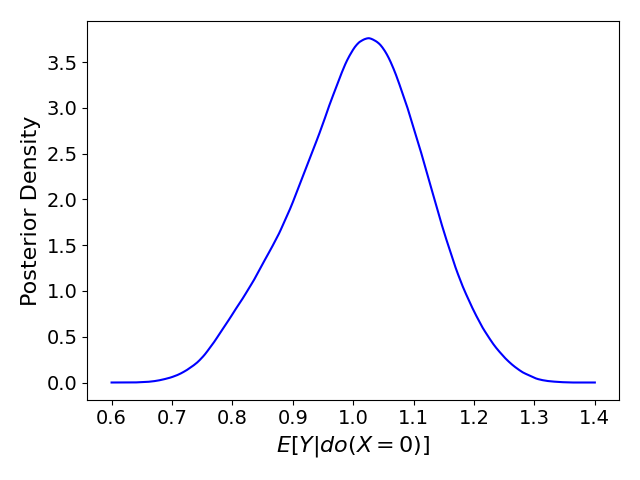}
  \caption{Posterior density of $E[Y | do(X = 0)]$ in Example 2.}
  \label{ex2plot}
\end{figure}
\section{Conclusion}
\label{conc}
We have proposed and demonstrated a novel method for sampling from a large class of interventional distributions, which can be used to find credible bounds on previously unmanageable causal quantities, including those with continuous and multivariate outcomes. We have demonstrated its use, and have also described how this approach can also be used for soft interventions.
\small
\bibliography{refs}

\begin{thebibliography}{15}
\providecommand{\natexlab}[1]{#1}
\providecommand{\url}[1]{\texttt{#1}}
\expandafter\ifx\csname urlstyle\endcsname\relax
  \providecommand{\doi}[1]{doi: #1}\else
  \providecommand{\doi}{doi: \begingroup \urlstyle{rm}\Url}\fi

\bibitem[Balke and Pearl(1997)]{doi:10.1080/01621459.1997.10474074}
Alexander Balke and Judea Pearl.
\newblock Bounds on treatment effects from studies with imperfect compliance.
\newblock \emph{Journal of the American Statistical Association}, 92\penalty0
  (439):\penalty0 1171--1176, 1997.
\newblock \doi{10.1080/01621459.1997.10474074}.
\newblock URL \url{https://doi.org/10.1080/01621459.1997.10474074}.

\bibitem[Bareinboim et~al.(2020)Bareinboim, Correa, Ibeling, and Icard]{pch}
E.~Bareinboim, J.~Correa, D.~Ibeling, and T.~Icard.
\newblock On pearl's hierarchy and the foundations of causal inference.
\newblock Technical Report R-60, Causal Artificial Intelligence Lab, Columbia
  University, Jul 2020.
\newblock In: ``Probabilistic and Causal Inference: The Works of Judea Pearl'',
  ACM Books, in press.

\bibitem[Chickering and Pearl(1997)]{cp1997}
D.M. Chickering and J.~Pearl.
\newblock A clinician's tool for analyzing non-compliance.
\newblock Technical Report R-241, UCLA Cognitive Systems Laboratory, 1997.

\bibitem[Cormen et~al.(2009)Cormen, Leiserson, Rivest, and Stein]{clrs}
Thomas~H. Cormen, Charles~E. Leiserson, Ronald~L. Rivest, and Clifford Stein.
\newblock \emph{Introduction to Algorithms, 3rd Edition}.
\newblock {MIT} Press, 2009.
\newblock ISBN 978-0-262-03384-8.

\bibitem[Correa and Bareinboim(2020)]{correa2020calculus}
J.~Correa and E.~Bareinboim.
\newblock A calculus for stochastic interventions: Causal effect identification
  and surrogate experiments.
\newblock In \emph{Proceedings of the 34th AAAI Conference on Artificial
  Intelligence}, New York, NY, 2020. AAAI Press.

\bibitem[Gelman et~al.(2013)Gelman, Carlin, Stern, Dunson, Vehtari, and
  Rubin]{gelman2013bayesian}
A.~Gelman, J.B. Carlin, H.S. Stern, D.B. Dunson, A.~Vehtari, and D.B. Rubin.
\newblock \emph{Bayesian Data Analysis, Third Edition}.
\newblock Chapman \& Hall/CRC Texts in Statistical Science. Taylor \& Francis,
  2013.
\newblock ISBN 9781439840955.

\bibitem[Huang and Valtorta(2006)]{10.5555/3020419.3020446}
Yimin Huang and Marco Valtorta.
\newblock Pearl's calculus of intervention is complete.
\newblock In \emph{Proceedings of the Twenty-Second Conference on Uncertainty
  in Artificial Intelligence}, UAI'06, page 217–224, Arlington, Virginia,
  USA, 2006. AUAI Press.
\newblock ISBN 0974903922.

\bibitem[Hyndman(1996)]{10.2307/2684423}
Rob~J. Hyndman.
\newblock Computing and graphing highest density regions.
\newblock \emph{The American Statistician}, 50\penalty0 (2):\penalty0 120--126,
  1996.
\newblock ISSN 00031305.
\newblock URL \url{http://www.jstor.org/stable/2684423}.

\bibitem[Kilbertus et~al.(2020)Kilbertus, Kusner, and
  Silva]{NEURIPS2020_e8b1cbd0}
Niki Kilbertus, Matt~J Kusner, and Ricardo Silva.
\newblock A class of algorithms for general instrumental variable models.
\newblock In H.~Larochelle, M.~Ranzato, R.~Hadsell, M.~F. Balcan, and H.~Lin,
  editors, \emph{Advances in Neural Information Processing Systems}, volume~33,
  pages 20108--20119. Curran Associates, Inc., 2020.

\bibitem[Pearl(2009)]{causality2e}
Judea Pearl.
\newblock \emph{Causality: Models, Reasoning and Inference}.
\newblock Cambridge University Press, USA, 2nd edition, 2009.
\newblock ISBN 052189560X.

\bibitem[Peters et~al.(2017)Peters, Janzing, and Sch\"olkopf]{eoci}
J.~Peters, D.~Janzing, and B.~Sch\"olkopf.
\newblock \emph{Elements of Causal Inference: Foundations and Learning
  Algorithms}.
\newblock MIT Press, Cambridge, MA, USA, 2017.

\bibitem[Samworth and Wand(2010)]{10.1214/09-AOS766}
R.~J. Samworth and M.~P. Wand.
\newblock {Asymptotics and optimal bandwidth selection for highest density
  region estimation}.
\newblock \emph{The Annals of Statistics}, 38\penalty0 (3):\penalty0 1767 --
  1792, 2010.
\newblock \doi{10.1214/09-AOS766}.
\newblock URL \url{https://doi.org/10.1214/09-AOS766}.

\bibitem[Tian and Pearl(2002)]{TianJin2002Agic}
Jin Tian and Judea Pearl.
\newblock A general identification condition for causal effects.
\newblock In \emph{Eighteenth national conference on artificial intelligence},
  pages 567--573. American Association for Artificial Intelligence, 2002.
\newblock ISBN 0262511290.

\bibitem[Zhang et~al.(2021)Zhang, Tian, and Bareinboim]{zhang2021pi}
J.~Zhang, J.~Tian, and E.~Bareinboim.
\newblock Partial identification of counterfactual distributions.
\newblock Technical Report R-78, Causal Artificial Intelligence Lab, Columbia
  University, Jun 2021.

\bibitem[Zhang and Bareinboim(2021)]{zhang2021bounding}
Junzhe Zhang and Elias Bareinboim.
\newblock Bounding causal effects on continuous outcome.
\newblock In \emph{Proceedings of the 35th AAAI Conference on Artificial
  Intelligence}, Vancouver, Canada, Feb 2021. AAAI Press.

\end{thebibliography}
\bibliographystyle{plainnat}
\end{document}